\documentclass[english]{article}
\usepackage{mathrsfs}
\usepackage{amsmath}
\usepackage{amssymb}
\usepackage{graphicx}
\usepackage{color} 

\definecolor{darkred}{rgb}{.7,0,0}

\definecolor{darkgreen}{rgb}{0,0.7,0}

\definecolor{darkblue}{rgb}{0,0,0.7}

\usepackage{enumerate}

\usepackage{amsthm}

\usepackage{mathrsfs}

\usepackage{amsfonts}

\usepackage{epsfig}

\usepackage{bm}

\usepackage{mathrsfs}

\usepackage{subfig}\usepackage[all]{xy}

\newcommand{\bbE}{\mathbb{E}}

\newcommand{\bbR}{\mathbb{R}}

\newcommand{\cO}{\mathcal{O}}

\newcommand{\testfunction}{\ensuremath{\mathcal{B}_b}}

\newtheorem{theorem}{Theorem}
\newtheorem{lemma}{Lemma}
\newtheorem{rem}{Remark}
\newtheorem{corollary}{Corollary}[section]
\newtheorem{prop}{Proposition}
\newcounter{hypA}
\newenvironment{hypA}{\refstepcounter{hypA}\begin{itemize}
  \item[({\bf A\arabic{hypA}})]}{\end{itemize}}

\title{Multilevel Sequential Monte Carlo Samplers for Normalizing Constants}

\author{Pierre Del  Moral\thanks{Center INRIA
Bordeaux Sud-Ouest \& Institut de Mathematiques de Bordeaux, Universite de Bordeaux I, 33405, FR}
\and
Ajay Jasra\thanks{Department of Statistics \& Applied Probability,
National University of Singapore, Singapore, 117546, SG}
\and
Kody Law\thanks{Computer Science and Mathematics Division,
Oak Ridge National Laboratory, 
Oak Ridge, 37934 TN, USA}
\and 
Yan Zhou\thanks{Department of Statistics \& Applied Probability,
National University of Singapore, Singapore, 117546, SG}}

\begin{document}
\maketitle

\begin{abstract}
This article considers the sequential Monte Carlo (SMC) approximation of ratios of normalizing constants associated to posterior
distributions which in principle rely on continuum models.  Therefore, the Monte Carlo estimation error
and the discrete approximation error must be balanced.  A multilevel strategy is utilized to substantially 
reduce the cost to obtain a given error level in the approximation as compared to standard estimators.
Two estimators are considered and relative variance bounds are given.
The theoretical results are numerically illustrated for the example of identifying a parametrized permeability 
in an elliptic equation given point-wise observations of the pressure.



\end{abstract}



 


\section{Introduction}

Over the past decades there has been an explosion of interest in accounting for uncertainty in the simulation
of systems in science and engineering applications which are governed by continuum limiting systems such 
as partial differential equations (PDEs) \cite{knio:10, walsh:86,walstrom:1967}.  
The setting bears similarities to the case of continuous stochastic processes, 
which have enjoyed attention for much longer (e.g.~\cite{oksendal:2003}).

Consider a sequence of probability measures $\{\eta_{l}\}_{l\geq 0}$ on a common measurable 
space $(E,\mathcal{E})$; assume that the probabilities have common dominating
finite-measure $du$. 
In particular, for some known $\kappa_{l}:E\rightarrow\mathbb{R}^+$, let
\begin{equation}
\label{eq:target}
\eta_{l}(du) = \frac{\kappa_l(u)}{Z_l}du
\end{equation}
where the normalizing constant $Z_l = \int_E\kappa_l(u)du$ may be unknown.  
The context of interest is when the sequence of densities is
associated to an `accuracy' parameter $h_l$, with $h_l\rightarrow 0$ 
as $l\rightarrow \infty$ with $\infty>h_0>h_1>\cdots>h_{\infty}=0$. 


When estimating statistics $\bbE_{\eta_{\infty}}[g(U)]$, for $g:E\rightarrow\mathbb{R}$,
in general one must approximate
the limiting measure by $\eta_L$ and perform statistical estimation with respect to this.  
For larger $L$, the  approximation of the limit is better, and yet the simulations are 
more expensive and indeed the measure may also be supported on a subspace of the
underlying space $E$ whose dimension is larger.  

Monte Carlo methods for statistical estimation are robust and scalable, although
they are plagued by a ``slow" convergence rate of $\cO(N^{-1/2})$ 
for approximations using $N$ degrees of freedom.  
Attempts to circumvent this issue, for example using sophisticated deterministic high-dimensional 
approximation methods typically result in some manifestation of the ``curse of dimensionality" \cite{bellman:1961}, 
although recent work has indicated potential for the mitigation of such effect for suitably regular problems \cite{griebel:04, schwab:11}.

The multilevel Monte Carlo framework  \cite{giles:08, giles:15,heinrich:01}
allows one to leverage in an optimal way 
the nested problems arising in this context, hence minimizing the necessary cost to
obtain a given level of mean square error.  
In particular, the  multilevel Monte Carlo method seeks to sample from
$\eta_0$ as well as 
a sequence of coupled pairs $(\eta_0,\eta_1),\dots, (\eta_{L-1},\eta_{L})$
and using a collapsing sum representation of $\bbE_{\eta_{L}}[g(U)]$. Then using a 
suitable trade off of computational effort, one can reduce the amount of work, relative to i.i.d.~sampling 
from $\eta_L$ and using Monte Carlo integration, for a given amount of error.
However, we are concerned with the scenario where such independent sampling is not possible,
that is, either $\eta_L$ or from the sequence of couples. As it is well-known, the use
of importance sampling to then use the collapsing sum representation, is often not reasonable,
in the sense that for importance proposals that can be sampled independently, the associated
variance typically explodes exponentially in the dimension of the problem (e.g.~\cite{bickel}).
As a result, there has been an extension of multilevel Monte Carlo methods 
in which the approximate target distribution can be sampled from directly,
to more sophisticated Monte Carlo techniques for inference; however, this is still in its infancy.  
Important examples include the preliminary exploration of multilevel Markov chain Monte Carlo (MCMC)
\cite{hoan:12, ketelsen2013hierarchical}, multilevel sequential Monte Carlo samplers \cite{besk:15}, 
multilevel ensemble Kalman filter \cite{ourmlenkf:15} and multilevel particle filters \cite{ourmlpf:15}.  
It should be noted that MCMC and SMC can perform at a polynomial cost in the dimension; see e.g.~\cite{beskos}
and the references therein.

A significant challenge for inference problems is estimation of the normalizing constant $Z_L$ or ratios thereof $Z_l/Z_k$, $L\geq l>k\geq 0$.
Such quantities are central to Bayesian model comparison and choice \cite{hoeting:99,wasserman:00}. In addition, obtaining unbiased estimates (in the sense that the
expectation is equal to the value, that is, potentially including discretization bias) are often central in
pseudo-marginal algorithms (e.g.~\cite{andrieu:09}). In general the calculation of these quantities are notoriously challenging
(see for instance \cite{yan}) from a computational perspective. 

In this article we extend the framework of \cite{besk:15}
to consider the estimation of the ratio of normalizing constants. 
This is a framework which uses SMC.
We consider both the `standard'  unbiased estimator (\cite{delm:04}) used in SMC, adapted to the multilevel setting
and an estimator which follows the collapsing sum approach for multilevel methods.
For the latter, we introduce a novel decomposition
of the normalizing constant of a Feynman-Kac formula, which corresponds to $Z_L/Z_0$, which facilitates
unbiased estimation. We consider new variance bounds for the estimator \cite{delm:04} and our new
estimate and show that, in general, both approaches perform in a similar manner. 
For a given level of error, the cost is less than a Monte Carlo estimate that uses i.i.d.~sampling from $\eta_{0}$,
to estimate $Z_L/Z_0$; we assume that the former is possible.


The paper is structured as follows.  In Section \ref{sec:est}
the setup will be given, along with a description of the multilevel algorithm and the 
new novel estimator for the normalizing constant.  
Section \ref{sec:theory} contains the theoretical results, including the main theorems of the paper
which allow the multilevel theory to carry through.
Finally, section \ref{sec:numerics} presents the results of numerical experiments on an
example Bayesian inverse problem. The proofs are housed in the appendix.

\section{Estimation}
\label{sec:est}

\subsection{Notations}

Let $(E,\mathcal{E})$ be a measurable space.
The notation $\testfunction(E)$ denotes the class of bounded and measurable real-valued functions. The 
supremum norm is written as $\|f\|_{\infty} = \sup_{u\in E}|f(u)|$ 
and $\mathcal{P}(E)$ is the set of probability measures on $(E,\mathcal{E})$. We will consider non-negative operators 
$K : E \times \mathcal{E} \rightarrow \bbR_+$ such that for each $u \in E$ the mapping $A \mapsto K(u, A)$ is a finite non-negative measure on $\mathcal{E}$ and for each $A \in \mathcal{E}$ the function $u \mapsto K(u, A)$ is measurable; the kernel $K$ is Markovian if $K(u, dv)$ is a probability measure for every $u \in E$.
For a finite measure $\mu$ on $(E,\mathcal{E})$,  and a real-valued, measurable $f:E\rightarrow\mathbb{R}$, we define the operations:
\begin{equation*}
    \mu K  : A \mapsto \int K(u, A) \, \mu(du)\ ;\quad 
    K f :  u \mapsto \int f(v) \, K(u, dv).
\end{equation*}
We also write $\mu(f) = \int f(u) \mu(du)$.

\subsection{Algorithm}

As described in the Introduction, the context of interest is when a sequence of densities 
$\{\eta_{l}\}_{l\ge 0}$, as in (\ref{eq:target}), are
associated to an `accuracy' parameter $h_l$, with $h_l\rightarrow 0$ 
as $l\rightarrow \infty$, such that $\infty>h_0>h_1\cdots>h_{\infty}=0$.  In practice one cannot treat $h_\infty=0$ and so must consider these distributions with $h_l>0$.
The laws with large $h_l$ are easy to sample from with low computational cost, but are very different from $\eta_{\infty}$, whereas, those distributions with small $h_l$ are
hard to sample with relatively high computational cost, but are closer to $\eta_{\infty}$. 
Thus, we choose a maximum level $L\ge 1$ and we will estimate
$$
\mathbb{E}_{\eta_L}[g(U)] := \int_E g(u)\eta_L(du)\ .
$$
By the standard telescoping identity used in MLMC, one has
\begin{align}
\mathbb{E}_{\eta_L}[g(U)] & =  \mathbb{E}_{\eta_0}[g(U)] + \sum_{l=1}^{L}\Big\{
\mathbb{E}_{\eta_l}[g(U)] - \mathbb{E}_{\eta_{l-1}}[g(U)]\Big\} \nonumber \nonumber \\ 
& =\mathbb{E}_{\eta_0}[g(U)] + \sum_{l=1}^{L}\mathbb{E}_{\eta_{l-1}}\Big[
\Big(\frac{\kappa_l(U)Z_{l-1}}{\kappa_{l-1}(U)Z_l} - 1\Big)g(U)\Big]\ .
\label{eq:ml_approx}
\end{align}

Suppose now that one applies an SMC sampler \cite{delm:06b} to obtain 
a collection of samples (particles) that sequentially approximate $\eta_0, \eta_1,\ldots, \eta_L$. 
We consider the case when one initializes the population of particles by sampling  i.i.d.~from $\eta_0$, then at every step  resamples and applies a MCMC Markov kernel to mutate the particles.
We denote by $(U_{0}^{1:N_0},\dots,U_{L-1}^{1:N_{L-1}})$, with $+\infty > N_0\geq N_1\geq \cdots  N_{L-1}\geq 1$, the samples after mutation; one resamples $U_l^{1:N_l}$ according to the weights $G_{l}(U_l^i) = 
(\kappa_{l+1}/\kappa_l)(U_l^{i})$, for indices $l\in\{0,\dots,L-1\}$.
We will denote by $\{M_l\}_{1\leq l\leq L-1}$ the sequence of MCMC kernels used at stages $1,\dots,L-1$, such that $\eta_{l}M_l = \eta_l$.
For $\varphi:E\rightarrow\mathbb{R}$, $l\in\{1,\dots,L\}$, we have the following estimator 
of $\bbE_{\eta_{l-1}}[\varphi(U)]$:
$$
\eta_{l-1}^{N_{l-1}}(\varphi) = \frac{1}{N_{l-1}}\sum_{i=1}^{N_{l-1}}\varphi(U_{l-1}^i)\ . 
$$
We define
$$
\eta_{l-1}^{N_{l-1}}(G_{l-1}M_l(du_l)) = \frac{1}{N_{l-1}}\sum_{i=1}^{N_{l-1}}G_{l-1}(U_{l-1}^i) M_l(U_{l-1}^i,du_l)\ .
$$
The joint probability distribution for the SMC algorithm is 
$$
\prod_{i=1}^{N_0} \eta_0(du_0^i) \prod_{l=1}^{L-1} \prod_{i=1}^{N_l} \frac{\eta_{l-1}^{N_{l-1}}(G_{l-1}M_l(du_l^i))}{\eta_{l-1}^{N_{l-1}}(G_{l-1})}\ .
$$
The algorithm is summarized in Figure \ref{tab:SMC}.
If one considers one more step in the above procedure, that would deliver samples 
$\{U_L^i\}_{i=1}^{N_L}$, a standard SMC sampler estimate of the quantity of interest in (\ref{eq:ml_approx})
is $\eta_L^{N}(g)$; the earlier samples are discarded. 
Within a multi-level context, a consistent SMC estimate of \eqref{eq:ml_approx}
is
\begin{equation}
\widehat{Y} =
\eta_{0}^{N_0}(g) + \sum_{l=1}^{L}\Big\{\frac{\eta_{l-1}^{N_{l-1}}(gG_{l-1})}{\eta_{l-1}^{N_{l-1}}(G_{l-1})} - \eta_{l-1}^{N_{l-1}}(g)\Big\}\label{eq:smc_est}\ , 
\end{equation}
The motivation for such a procedure is that, as shown in \cite{besk:15}, the amount of work, for a given level of error, relative to i.i.d.~sampling
from $\eta_L$ is reduced. Thus the idea of using the approach is clear. However, as is well known in the literature (e.g.~\cite{delm:06b}) SMC samplers can
also estimate ratios of normalizing constants as a by-product of the algorithm. We now consider this and also the amount of work to obtain
a given level of error in this context.

\begin{figure}[!h]
\begin{flushleft}
\medskip
\hrule
\medskip
{\itshape
\begin{enumerate}
\item[\textit{0.}] Sample $U_0^1,\dots U_0^{N_0}$ i.i.d.\@ from $\eta_0$ and compute $G_0(u_0^i)$ for each sample $i\in\{1,\dots,N_0\}$:
%
%
Set $l=0$.
\vspace{0.1cm}
\item[\textit{1.}]  Sample $\check{U}_{l}^{1},\dots,\check{U}_{l}^{N_{l+1}}$  with replacement from $u_{l}^{1:N_l}$ with selection probabilities
$\{G_{l}(u_l^1)/\sum_{j=1}^{N_l} G_l(u_l^j),$ $\dots,(G_{l}(u_l^{N_l})/\sum_{j=1}^{N_l} G_l(u_l^j)\}$.\\
\item[\textit{2.}] Sample $U_{l+1}^i |\check{u}_{l}^{i}$ from $M_{l+1}(\check{u}_{l}^{i},\cdot)$  and compute $G_{l+1}(u_{l+1}^i)$ for each sample $i\in\{1,\dots,N_{l+1}\}$.\\
\item[\textit{3.}] Set $l=l+1$. If $l=L$ stop, otherwise return to the start of Step 1.
%
%
\end{enumerate} }
\medskip
\hrule
\medskip
\end{flushleft}
\vspace{-0.5cm}
\caption{The SMC algorithm.}
\label{tab:SMC}
\end{figure}

\subsection{Normalizing Constant}

Define, for $l\geq 0$
$$
\gamma_l(du_l) = \int_{E^l} \Big(\prod_{p=0}^{l-1} G_{p}(u_p)\Big)\eta_0(du_l) 
\prod_{p=1}^{l} M_p(u_{p-1},du_p).
$$
In our context, it is well known that:
$$
\gamma_l(1) = \frac{Z_l}{Z_0} = \prod_{p=0}^{l-1} \eta_p(G_p).
$$
This suggests the estimator:
\begin{equation}
\gamma_l^{N_{0:l-1}}(1) = \prod_{p=0}^{l-1} \eta_p^{N_p}(G_p) \label{eq:nc_est_stand}
\end{equation}
which is known to be unbiased (\cite{delm:04}). We consider the relative variance
of this estimator in Section \ref{sec:theory}. However, at least on appearance
it may not take advantage of the nature of the ML method. In addition, we 
show that the new estimator below, can potentially be leveraged to remove
the discretization bias.

We propose the following procedure. It should be remarked that it holds in the particular
context under study, but not for general Feynman-Kac models as will be explained below.
We have that for any $l\geq 2$
\begin{eqnarray*}
\gamma_l(1) & = & \eta_0(G_0) + \sum_{p=2}^l \Big(\gamma_p(1) - \gamma_{p-1}(1)\Big)\\
& = & \eta_0(G_0) + \sum_{p=2}^l \Big(\gamma_{p-2}\big(G_{p-2}(M_{p-1}(G_{p-1}) - 1)\big)\Big) \\
& = & \eta_0(G_0) + \sum_{p=2}^l \Big(\gamma_{p-2}\big(G_{p-2}(G_{p-1} - 1)\big)\Big).
\end{eqnarray*}
It is the final line that we will approximate with our MLSMC sampler. It is noted that the final
line holds in the specific case of interest, but is not generally true for a given Feynman-Kac formula. 
The proposed approximation is
$$
\tilde{\gamma}_l^{N_{0:l-2}}(1) = \eta_0^{N_0}(G_0) + \sum_{p=2}^l \Big(\gamma_{p-2}^{N_{0:p-2}}\big(G_{p-2}(G_{p-1} - 1)\big)\Big)
$$
where for any $g\in\mathcal{B}_b(E)$, $p\geq 2$
$$
\gamma_{p-2}^{N_{0:p-2}}(g) = \Big(\prod_{k=0}^{p-3}\eta_{k}^{N_k}(G_k)\big)\eta_{p-2}^{N_{p-2}}(g).
$$
Note that for $l\geq 2$, one has, almost surely,
$$
\tilde{\gamma}_l^{N_{0:l-1}}(1) \neq \prod_{p=0}^{l-1} \eta_p^{N_p}(G_p).
$$
Using \cite{delm:04} it clearly follows that 
$$
\gamma_l(1) = \mathbb{E}[\tilde{\gamma}_l^{N_{0:l-2}}(1)]
$$
where $\mathbb{E}$ is the expectation w.r.t.~the law of the SMC algorithm; the estimator is unbiased.

\subsection{Biased Estimator}

Noting the estimator \eqref{eq:smc_est} another alternative estimator of $\gamma_l(1)$
is
$$
\prod_{p=0}^{l-1} \Bigg(
\eta_{0}^{N_0}(G) + \sum_{l=1}^{p}\Big\{\frac{\eta_{l-1}^{N_{l-1}}(G_pG_{l-1})}{\eta_{l-1}^{N_{l-1}}(G_{l-1})} - \eta_{l-1}^{N_{l-1}}(G_p)\Big\}
\Bigg).
$$
One can easily prove that this estimate is consistent, but biased, in the sense
that
$$
\mathbb{E}\Bigg[\prod_{p=0}^{l-1} \Bigg(
\eta_{0}^{N_0}(G) + \sum_{l=1}^{p}\Big\{\frac{\eta_{l-1}^{N_{l-1}}(G_pG_{l-1})}{\eta_{l-1}^{N_{l-1}}(G_{l-1})} - \eta_{l-1}^{N_{l-1}}(G_p)\Big\}
\Bigg)\Big]
\neq \gamma_{l}(1).
$$
However, the main reason why one may not want to consider its use is due to the 
cost of computing this estimate. If $\sum_{p=0}^{l-1} N_p C_p$ is the ordinary cost of computing \eqref{eq:nc_est_stand} ($C_p$ is the cost per sample), then the cost of this estimator is $\sum_{p=0}^{l-1} N_p\sum_{q=p}^{l-1} C_q$. Such a procedure is undesirable in general
and this is not investigated further.

\subsection{Estimator with no Discretization Bias}

Let $M\in\{1,2,\dots\}$ be a random variable that is independent of the MLSMC
algorithm with $\mathbb{P}_M(M\geq m)>0~\forall m>0$. Suppose further that one can
prove for $N_0, N_1,\dots$ fixed that
\begin{eqnarray}
\lim_{p\rightarrow\infty} \mathbb{E}\Big[\Big(
\gamma_{p-2}^{N_{0:p-2}}\big(G_{p-2}G_{p-1})-
\gamma_{\infty}(1)\Big)^2
\Big]^{1/2} & = & 0 \label{eq:glynn1} \\
\lim_{p\rightarrow\infty} \mathbb{E}\Big[\Big(
\gamma_{p-2}^{N_{0:p-2}}\big(G_{p-2})-
\gamma_{\infty}(1)\Big)^2
\Big]^{1/2} & = & 0 \label{eq:glynn3} \\
\sum_{p=2}^{\infty}\frac{1}{\mathbb{P}_M(M\geq p)}
\mathbb{E}\Big[\Big(
\gamma_{p-2}^{N_{0:p-2}}\big(G_{p-2}(G_{p-1})-
\gamma_{\infty}(1)\Big)^2
\Big]
 & < & +\infty \label{eq:glynn2}
\end{eqnarray}
then one can use the estimator from \cite{rg:15} to obtain an unbiased estimator for $\gamma_{\infty}(1)$:
$$
\frac{1}{\mathbb{P}_M(M\geq 1)}\eta_0^{N_0}(G_0) + \sum_{p=2}^M
\frac{1}{\mathbb{P}_M(M\geq p)}
 \Big(\gamma_{p-2}^{N_{0:p-2}}\big(G_{p-2}(G_{p-1} - 1)\big)\Big).
$$
Note that, even if one can prove \eqref{eq:glynn1}-\eqref{eq:glynn2},
one must be prepared to spend an arbitrary amount of computational cost,
which is not reasonable in the current context. Hence we do not consider this
further here. We further remark that this particular approach is unlikely to work when
estimating $\mathbb{E}_{\eta_{\infty}}[g(U)]$ (as in \eqref{eq:ml_approx}) as there is no unbiased property of the estimators
(unbiased in the sense of expectations and not associated to the discretization).

\section{Theory}\label{sec:theory}

\subsection{Relative Variance Bounds}

Throughout $E$ is compact.
We make the following assumptions:

\begin{hypA}
\label{hyp:A}
There exist $0<\underline{C}<\overline{C}<+\infty$ such that
\begin{eqnarray*}
\sup_{0\leq l< \infty} \sup_{u\in E} \kappa_l (u) & \leq & \overline{C}\ ;\\
\inf_{0\leq l< \infty} \inf_{u\in E} \kappa_l (u) & \geq & \underline{C}\ .
\end{eqnarray*}
\end{hypA}

\begin{hypA}
\label{hyp:B}
There exists a $\rho\in(0,1)$ such that for any $l\ge 1$, $(u,v)\in E^2$, $A\in\mathcal{E}$:
$$
\int_A M_l(u,du') \geq \rho \int_A M_l(v,dv')\ .
$$
\end{hypA}

These assumptions are almost identical to those in \cite{besk:15}. (A\ref{hyp:A}) 
is different but equivalent to (A1) in \cite{besk:15}. The proofs of the following Theorems are
in Appendices  \ref{app:proof_new} and \ref{app:proof} respectively. It is remarked that there
are other results in the spirit of Theorem \ref{theo:2} below, (see \cite{cdg:11,schweizer:12}) but the bounds are not sharp enough for
the purposes of this work.


\begin{theorem}\label{theo:2}
Assume (A\ref{hyp:A}-\ref{hyp:B}). Then there exists a $c,C<+\infty$ such that
for any $L\geq 2$, $N_0 \geq N_1 \geq \dots \geq N_{L-1} > cL$,
$$
\mathbb{E}\Big[\Big(\frac{\gamma_L^{N_{0:L-1}}(1)}{\gamma_L(1)} - 1 \Big)^2\Big] \leq
$$
$$
C\sum_{p=0}^{L-1} \frac{1}{N_p}\Big(
\Big(\sum_{q=p}^{L-1}  
\Big\|\frac{G_q}{\eta_q(G_q)}-1\Big\|_{\infty}
\Big)^2 +
\Big\|\frac{G_p}{\eta_p(G_p)}-1\Big\|_{\infty}\frac{(p+1)}{N_p}
\Big).
$$
\end{theorem}


\begin{theorem}\label{theo:1}
Assume (A\ref{hyp:A}-\ref{hyp:B}). Then there exists a $c,C<+\infty$ such that
for any $L\geq 2$, $N_0 \geq N_1 \geq \dots \geq N_{L-2} > c(L-1)$,
$$
\mathbb{E}\Big[\Big(\frac{\tilde{\gamma}_L^{N_{0:L-2}}(1)}{\gamma_L(1)} - 1 \Big)^2\Big] \leq
$$
$$
C\Big(\frac{1}{N_0} + \sum_{p=2}^L \frac{(p-1)}{N_{p-2}}\|G_{p-1}-1\|_{\infty}^2
+ \sum_{p=2}^L \sum_{q=2}^{p-1} \frac{(q-1)}{N_{q-2}}
\|G_{p-1}-1\|_{\infty}\|G_{q-1}-1\|_{\infty}
\Big).
$$
\end{theorem}

\subsection{Cost Analysis}

In order to investigate the cost for a given level of error, we introduce the following assumption.

\begin{hypA}
\label{hyp:C}
(i) There exist $\alpha, \zeta>0$, 
and a $C>0$ such that for all $p>0$
\begin{equation}\label{eq:mlrates}
\begin{cases}
|\frac{\gamma_p(1)}{\gamma_\infty(1)}-1 | & \leq C h_p^\alpha ;\\
{\rm C} (G_{p-1}) & \leq C h_p^{-\zeta},
\end{cases}
\end{equation}
where ${\rm C} (G_{p-1})$ denotes the cost to evaluate $G_{p-1}$.

(ii) There exist a $\beta>0$ and a $C>0$ such that for all $p>0$
$$
\Big\|\frac{G_{p-1}}{\eta_{p-1}(G_{p-1})}-1\Big\|_{\infty}^2  \leq C h_p^\beta.
$$

(iii) There exist a $\beta>0$ and a $C>0$ such that for all $p>0$
$$
\|G_{p-1}-1\|_{\infty}^2  \leq C h_p^\beta.
$$
\end{hypA}

\begin{corollary}\label{col:cost}
Assume (A\ref{hyp:A},\ref{hyp:B},\ref{hyp:C}(i)(ii)) and $2\alpha\geq {\rm max}\{\beta,\zeta\}$. 
Then for any $\varepsilon>0$, there exist $L,\{N_l\}_{l=0}^L$
and $C>0$ such that 
\begin{equation}
\frac1{\gamma_\infty(1)^2}\mathbb{E}\Big[\Big({\gamma_L^{N_{0:L-2}}(1)}-\gamma_\infty(1)\Big)^2\Big] \leq C \varepsilon^2,
\label{eq:mlmse}
\end{equation}
for the following cost 
\begin{equation}\label{eq:mlncCosts}
{\rm COST} \leq C 
\begin{cases}
\varepsilon^{-2}|\log(\varepsilon)|, & \text{if} \quad \beta > \zeta,\\              
\varepsilon^{-2} |\log(\varepsilon)|^3, & \text{if} \quad \beta = \zeta,\\
\varepsilon^{-\left( 2 + \frac{\zeta-\beta}{\alpha} \right)}|\log(\varepsilon)|, & \text{if} \quad \beta < \zeta. 
\end{cases}
\end{equation}
\end{corollary}

\begin{corollary}\label{col:cost1}
Assume (A\ref{hyp:A},\ref{hyp:B},\ref{hyp:C}(i)(iii)) and $2\alpha\geq {\rm max}\{\beta,\zeta\}$. 
Then for any $\varepsilon>0$, there exist $L,\{N_l\}_{l=0}^L$
and $C>0$ such that 
\begin{equation}
\frac1{\gamma_\infty(1)^2}\mathbb{E}\Big[\Big({\tilde{\gamma}_L^{N_{0:L-2}}(1)}-\gamma_\infty(1)\Big)^2\Big] \leq C \varepsilon^2,
\label{eq:mlmse1}
\end{equation}
for the following cost 
\begin{equation}\label{eq:mlncCosts1}
{\rm COST} \leq C 
\begin{cases}
\varepsilon^{-2}|\log(\varepsilon)|, & \text{if} \quad \beta > \zeta,\\              
\varepsilon^{-2} |\log(\varepsilon)|^3, & \text{if} \quad \beta = \zeta,\\
\varepsilon^{-\left( 2 + \frac{\zeta-\beta}{\alpha} \right)}|\log(\varepsilon)|, & \text{if} \quad \beta < \zeta.
\end{cases}
\end{equation}
\end{corollary}

We give the proof for Corollary \ref{col:cost1} only. The proof of  Corollary \ref{col:cost} is almost identical. The only
difference is treating the term in the relative variance of 
$$
\sum_{p=0}^{L-1} \Big\|\frac{G_p}{\eta_p(G_p)}-1\Big\|_{\infty}\frac{(p+1)}{N_p^2}
$$
which is smaller than $\mathcal{O}(\epsilon^2)$, under our assumptions.

\begin{proof}[Proof of Corollary \ref{col:cost1}]
The MSE can be bounded by
$$
\frac1{\gamma_{\infty}(1)^2}\mathbb{E}\Big[\Big({\tilde{\gamma}_L^{N_{0:L-2}}(1)}-\gamma_{\infty}(1)\Big)^2\Big] \leq 
$$
$$
\left(\frac{\gamma_L(1)}{\gamma_{\infty}(1)}\right)^2\mathbb{E}\Big[\Big(\frac{\tilde{\gamma}_L^{N_{0:L-2}}(1)}{\gamma_L(1)}-1\Big)^2\Big] 
+ \left |\left(\frac{\gamma_L(1)}{\gamma_{\infty}(1)} - 1\right)\right|^2 .
$$
Following from (A\ref{hyp:C}(i)), 
the second term requires that $h_L^\alpha \eqsim \varepsilon$, and assuming
$h_L = M^{-L}$ for some $M\geq2$, this translates to $L\eqsim \log\varepsilon$.  
Notice that it also follows that $\left(\frac{\gamma_L(1)}{\gamma_{\infty}(1)}\right)^2 = \cO(1)$.
Now, defining $V_p = \|G_{p-1}-1\|_{\infty}^2$, Theorem \ref{theo:1} provides the following bound for the 
first term 
$$
\mathbb{E}\Big[\Big(\frac{\tilde{\gamma}_L^{N_{0:L-2}}(1)}{\gamma_L(1)}-1\Big)^2\Big] \leq V 
:= C\left(\frac{1}{N_0} + L \sum_{p=1}^{L-1} \frac{V_p}{N_{p-1}}\right).  
$$
To see this observe that 
$$
\sum_{p=1}^{L-1} \sum_{q=1}^{p} \frac{q}{N_{q-1}}V_p^{1/2}V_q^{1/2} 
= \sum_{p=1}^{L-1} \frac{p}{N_{p-1}} V_p^{1/2} \sum_{q=p}^{L-1} V_q^{1/2} 
\leq C L \sum_{p=1}^{L-1} \frac{V_p}{N_{p-1}}.
$$
Optimizing the cost, given that the variance is $\mathcal{O}(\varepsilon^2)$, dictates that $N_l \propto \sqrt{LV_l/C_l} \eqsim L^{1/2} h_l^{(\beta+\zeta)/2}$.
The constraint then requires that $N_l\propto L\varepsilon^{-2}K_L h_l^{(\beta+\zeta)/2}$, 
where $K_L = \sum_{l=1}^{L-1} h_l^{(\beta-\zeta)/2}$. 
By assumption ${\rm max}\{\beta,\zeta\} \leq 2 \alpha$, so $(\beta+\zeta)/2\alpha \leq 2$
and the 
requirement for all the $N_l$ in Theorem \ref{theo:1} is guaranteed
(as long as the proportionality constant is greater than 1). 
Therefore, the MSE is controlled by $\cO(\varepsilon^2)$ with a cost given by 
$$
\sum_{l=0}^L N_l C_l \eqsim L \varepsilon^{-2} K_L^{2} \, ,
$$
and the result follows.

\end{proof}

\begin{rem}
If one were able to perform i.i.d.~sampling from $\eta_0$ (denote the samples $u^1,\dots,u^N$), with estimator
$$
\frac{1}{N}\sum_{i=1}^N \frac{\gamma_L(u^i)}{\gamma_0(u^i)}
$$
for $Z_L/Z_0$
a computational effort proportional to $N h_L^{-\zeta}$ is used,
with $N$ the number of simulated samples. 
To make the overall error (bias squared plus variance) of using i.i.d.~sampling
$\mathcal{O}(\epsilon^2)$ then one must take $N\propto\mathcal{O}(\epsilon^{-2})$, as the variance of the MC estimate is $\mathcal{O}(N^{-1})$, independently of $L$. This is a computational cost of $\mathcal{O}(\epsilon^{-2}h_L^{-\zeta})$ is used which is far worse than MLSMC samplers in most cases of practical
interest. 
\end{rem}


\section{Numerical Example}
\label{sec:numerics}

\subsection{Setup}

The performance of the proposed estimator will be demonstrated by a Bayesian
inverse problem example. The same example was also used in \cite{besk:15},
which introduced the MLSMC algorithm.

Introduce the Gelfand triple $V := H^1(D) \subset L^2(D) \subset H^{-1}(D) =:
V^{*}$, where the domain $D$ will be understood.  Let $D \subset \bbR^d$ with
$\partial D\in C^1$ convex. For $f \in V^{*}$, consider the following PDE on
$D$:
\begin{align}
  -\nabla \cdot (\hat{u}\nabla p) &= f, \qquad\text{on }D,
  \label{eq:elliptic} \\
  p &= 0, \qquad \text{on }  \partial D,
\end{align}
where
\begin{equation}
  \hat{u}(x) = \bar{u}(x) + \sum_{k = 1}^K u_k \sigma_k \phi_k(x).
\end{equation}
Define $u = \{u_k\}_{k = 1}^K$, with $u_k \stackrel{\textrm{i.i.d.}}{\sim} \mathcal{U}[-1,1]$ (the uniform distribution on [-1,1]). 
Assume that
$\bar{u}, \phi_k\in C^{\infty}$ for all $k$ and $\|\phi_k\| = 1$. In
particular $\{\sigma_k\}_{k=1}^K$ decay with $k$. In addition, the following
property shall hold:
\begin{equation}
  \inf_{x}\hat{u}(x) \ge \inf_{x}\bar{u}(x) - \sum_{k = 1}^K \sigma_k \ge u_{*}
  > 0
\end{equation}
so that the operator on the left-hand side of Equation~\eqref{eq:elliptic} is
uniformly elliptic. Let $p(\cdot;u)$ denote the weak solution of
Equation~\eqref{eq:elliptic} for parameter $u$. Define the following
vector-valued function
\begin{equation*}
  \mathcal{G}(p) = [g_1(p),\dots,g_M(p)]^\mathsf{T},
\end{equation*}
where $g_m$ are elements of the dual space $V^*$ for $m = 1,\dots,M$. It is
assumed that the data take the form
\begin{equation}
  y = \mathcal{G}(p) + \xi, \qquad \xi\sim\mathcal{N}(0,\Xi),
\end{equation}
where $\mathcal{N}(0,\Xi)$ denotes the Normal distribution with zero mean
and covariance $\Xi$.

The specific setting of the simulations are as the following. Let $D = [0,1]$
and $f(x) = 100x$. Set $K = 50$, $\bar{u}(x) = 0.15 = \text{const.}$, $\sigma_k
= (2/5)4^{-k}$ $\phi_k(x) = \sin(k\pi x)$ if $k$ is odd and $\phi_k(x) =
\cos(k\pi x)$ if $k$ is even. The forward problem at resolution level $l$ is
solved using a finite element method with piecewise linear shape functions on a
uniform mesh of with $h_l = 2^{-(l + k)}$, for some starting $k \ge 1$ (so that
there are at least two grid-blocks in the coarsest, $l = 0$, case). Thus, on
level $l$ the finite element basis functions are
$\{\psi_i^l\}_{i=1}^{2^{l+k}-1}$ defined as (for $x_i = i\cdot2^{-(l + k)}$):
\begin{equation*}
  \psi_i^l(x) =
  \begin{cases}
    (1/h_l)[x - (x_i - h_l)] \qquad \text{if } x\in[x_i - h_l, x_i], \\
    (1/h_l)[(x_i + h_l) - x] \qquad \text{if } x\in[x_i, x_i + h_l].
  \end{cases}
\end{equation*}
The function of interest $g$ is taken as the solution of the forward problem at
the midpoint of the domain, that is $g(u) = p(0.5;u)$. The observation operator
is $\mathcal{G}(u) = [p(0.25;u),p(0.75;u)]^\mathsf{T}$, and the observational
noise covariance is taken to be $\Xi = 0.25^2I$. 

Detailed error rates analysis of this example can be found in \cite{besk:15}.
In particular, when the purpose of the study was to estimate $\eta_L(g)$, the
variance rate was $\beta = 4$ empirically. Later we will show that for
estimating the normalizing constant, the variance rate is very similar.

\subsection{Verification of Assumptions}

Assumptions (A\ref{hyp:A}) and (A\ref{hyp:C}(i)(iii)) (for $|\frac{\gamma_p(1)}{\gamma_\infty(1)}-1 |$), with $\beta=2\alpha=2$,
follow from Proposition 4.1 of \cite{besk:15}.  For (A\ref{hyp:C}(ii)) this follows directly from proving (A\ref{hyp:C}(iii)).
It is natural to model the cost
at level $p$ by a power of the number degrees of freedom, which is in turn
related to $h_p^{-1}$, verifying (A\ref{hyp:C}(i)) (for ${\rm C} (G_{p-1})$).  The stiffness matrix of
the finite element method is tridiagonal and thus the system can be solved with
cost $\mathcal{O}(2^{l + k})$, corresponding to a computational cost rate of
$\zeta = 1$.  Assumption (A\ref{hyp:B}) is verified for Gibbs sampler in
section 4.2 of \cite{besk:15}.

\subsection{Experiments}

We begin by using the theoretical rates $\beta = 2\alpha = 2$ to estimate the MSE
and hence the cost ratio. Three cases are considered:
\begin{itemize}
\item{A standard SMC
algorithm, with the estimator $\gamma_l^{N_{0:l-1}}(1)$.}
\item{MLSMC sampler for $\gamma_l^{N_{0:l-1}}(1)$.}
\item{MLSMC sampler for $\tilde{\gamma}_l^{N_{0:l-2}}(1)$.}
\end{itemize}
The cost vs. MSE
is plotted in Figure~\ref{fig:elliptic_ncl}. The cost rates are $-1.271$,
$-0.967$, and $-1.038$ for the SMC, MLSMC with the standard estimator, and
MLSMC with the new estimator, respectively. It is clear that the MLSMC
algorithm with both estimators provides superior performance when compared to
the standard SMC algorithm. It is interesting that for the given MLSMC ensemble, 
the performance of the new estimator is comparable
to that of the standard estimator, as proven in Corollaries \ref{col:cost} and \ref{col:cost1}. 
It shall be noted that in practice, given the same samples ($U_0^{1:N_0},\dots,U_{L - 1}^{1:N_{L-1}})$, 
the new estimator is capable of estimating $\gamma_{L + 1}(1)$ while the standard one can only
estimate the $\gamma_L(1)$, which has a higher bias.

\begin{figure}
  \includegraphics[width=\linewidth]{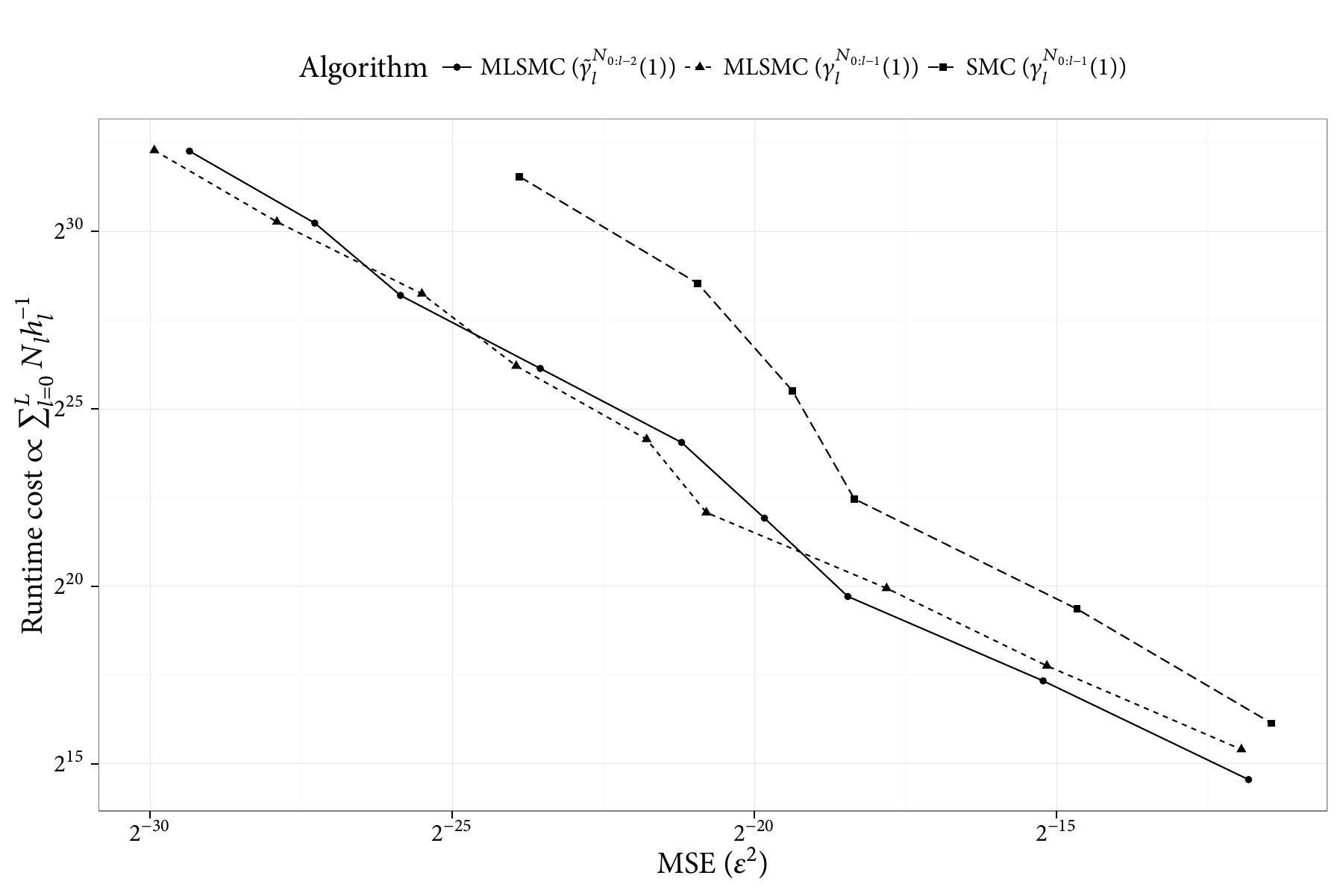}
  \caption{Computational cost against mean squared error}
  \label{fig:elliptic_ncl}
\end{figure}

The variance rate $\beta$ can also be estimated empirically by consider the
variance of $\eta_l(G_l)$. The quantity, multiplied by the sample size, as a
proxy of $V_l$ is plotted in Figure~\ref{fig:elliptic_ncv}. The estimated
empirical rate is $\beta = 4.148$. This is consistent with the rate estimates
in \cite{besk:15}.

\begin{figure}
  \includegraphics[width=\linewidth]{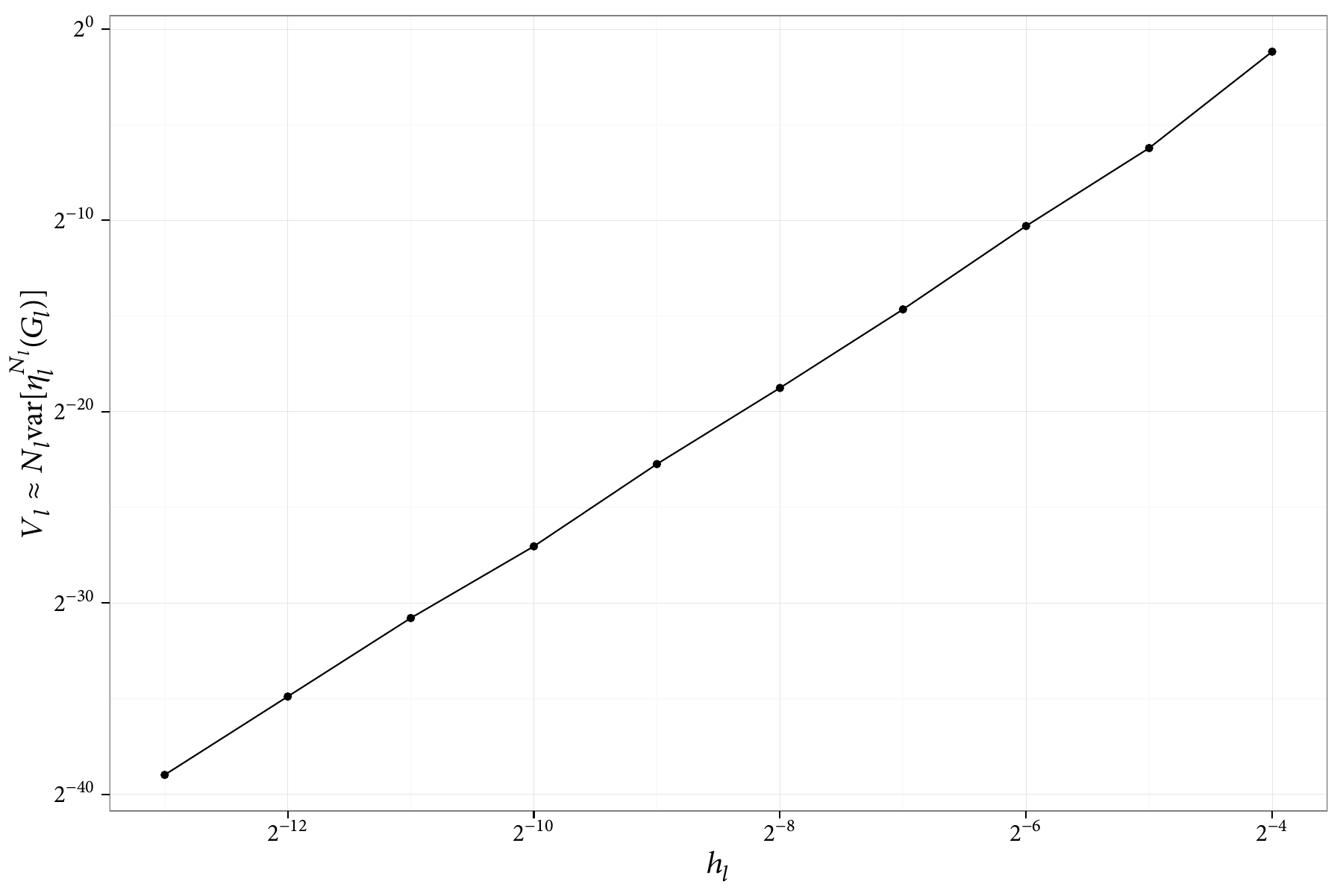}
  \caption{Variance rate estimate}
  \label{fig:elliptic_ncv}
\end{figure}

\subsubsection*{Acknowledgements}  KJHL was supported in part by DARPA FORMULATE 
and in part by ORNL LDRD Strategic Hire.
AJ \& YZ were supported by Ministry of Education AcRF tier 2 grant, R-155-000-161-112.

\appendix

\section{Notations}

We give a collection of defintions which are used in the appendices.
Let $n\geq 0$, $F\in\mathcal{B}_b(E\times E)$ and define
$$
(\gamma_n^{N_{0:n}})^{\otimes 2}(F) = \Big(\prod_{p=0}^{n-1}\eta_p^{N_p}(G_p)\Big)^2
(\eta_n^{N_n})^{\otimes 2}(F)
$$
where for a finite (possibily signed) measure on $E$, $\mu$, $\mu^{\otimes 2}(d(u_1,u_2)) = \mu(du_1)\mu(du_2)$.
We recall the semi-group for $p\leq n$ (for $p=n$ it is the identity operator):
$$
Q_{p,n}(x_p,dx_n) = \int_{E^{n-p-1}} Q_{p+1}(x_{p},dx_{p+1})\dots Q_n(x_{n-1},dx_n)
$$
where for $n\geq 1$, $Q_n(x,dy) = G_{n-1}(x) M_n(x,dy)$. We also define the coalescent operator for $F\in\mathcal{B}_b(E\times E)$, $(x,y)\in E\times E$:
$$
C(F)(x,y) = F(x,x).
$$
Then for $0\leq s \leq (n+1)$, $0\leq i_1< \cdots < i_n\leq n$, $F\in\mathcal{B}_b(E\times E)$
$$
\Gamma_n^{i_1:i_s}(F) = \gamma_{i_1}^{\otimes 2}CQ_{i_1,i_2}^{\otimes 2}CQ_{i_2,i_3}^{\otimes 2}\dots C Q_{i_s,n}^{\otimes 2}(F)
$$
and 
$$
\overline{\Gamma}_n^{i_1:i_s}(F) = \frac{1}{\gamma_n(1)^2}\Gamma_n^{i_1:i_s}(F).
$$
The conventions, for $s=0$, $\Gamma_n^{\emptyset}(F) = \gamma_n^{\otimes 2}(F)$ and 
$\overline{\Gamma}_n^{\emptyset}(F) = \eta_n^{\otimes 2}(F)$ are adopted.
Recall the selection-mutation operator for any $\mu\in\mathcal{P}(E)$, $n\geq 1$
$$
\Phi_n(\mu)(dx) = \frac{\mu(G_{n-1}M_n(\cdot,dx))}{\mu(G_{n-1})}.
$$
$\mathscr{F}_{n}^{N_{0:n}}$ denotes the natural filtration generated by the particle system up-to time $n$.
For $f_1,f_2\in\mathcal{B}_b(E)$ we write the tensor product of functions for every $(x,y)\in E\times E$:
$$
f_1\otimes f_2(x,y) = f_1(x)f_2(y).
$$

\section{Proofs for Theorem \ref{theo:2}}\label{app:proof_new}

\begin{lemma}\label{lem:new}
Assume (A\ref{hyp:A}-\ref{hyp:B}). Then there exist a $C<+\infty$ such that for any $0\leq p \leq n$, $x\in E$:
$$
\Big|\frac{Q_{p,n}(1)(x)}{\prod_{q=p}^{n-1}\eta_q(G_q)}-1
\Big| \leq C \sum_{q=p}^{n-1} \Big\|\frac{G_q}{\eta_q(G_q)}-1\Big\|_{\infty}
$$
\end{lemma}

\begin{proof}
We fix $n,p$ and note that the case $p=n$ is trivial, so we suppose $p<n$.
We prove the result by induction. We consider $p=n-1$ and thus
$$
\frac{Q_{p,n}(1)(x)}{\prod_{q=p}^{n-1}\eta_q(G_q)}-1 = \frac{G_{n-1}(x)}{\eta_{n-1}(G_{n-1})}-1
$$
so the initialization is proved. Suppose the result holds at rank $p$ and consider the case $p-1$. We have
$$
\frac{Q_{p-1,n}(1)(x)}{\prod_{q=p-1}^{n-1}\eta_q(G_q)}-1 = 
$$
$$
\Big(\frac{G_{p-1}(x)}{\eta_{p-1}(G_{p-1})} - 1\Big)M_p\Big(\frac{Q_{p,n}(1)(x)}{\prod_{q=p}^{n-1}\eta_q(G_q)}\Big)(x) + 
M_p\Big(\frac{Q_{p,n}(1)(x)}{\prod_{q=p}^{n-1}\eta_q(G_q)}-1\Big)(x).
$$
By \cite[Lemma 4.1]{cdg:11}
\begin{equation}
\frac{Q_{p,n}(1)(x)}{\prod_{q=p}^{n-1}\eta_q(G_q)}
\leq C\label{eq:h_cont}
\end{equation}
where $C$ does not depend upon $p,n$. Thus, by applying the induction hypothesis and the above result it follows that:
$$
\Big|\frac{Q_{p-1,n}(1)(x)}{\prod_{q=p-1}^{n-1}\eta_q(G_q)}-1
\Big| \leq C \sum_{q=p-1}^{n-1} \Big\|\frac{G_q}{\eta_q(G_q)}-1\Big\|_{\infty}
$$
and hence the proof is completed.
\end{proof}

The result below follows one in \cite{cdg:11}.

\begin{prop}\label{prop:1}
Assume (A\ref{hyp:A}-\ref{hyp:B}). Then there exists a $C<+\infty$ such that for any $n\geq 0$, $F\in\mathcal{B}_b(E\times E)$ and $N_0\geq\cdots\geq N_n > c(n+1)$
$$
\Big|
\mathbb{E}\Big[\frac{(\gamma_n^{N_{0:n}})^{\otimes 2}(F)}{\gamma_n(1)^2}\Big] - 
\eta_n^{\otimes 2}(F)
\Big| \leq 8c\|F\|_{\infty}\sum_{p=0}^n \frac{1}{N_p}.
$$
\end{prop}

\begin{proof}
The case with $F$ constant essentially follows from the proofs of \cite{cdg:11}. The only
difference is the fact that we have a decreasing number of samples; 
this does not change the calculations of that paper, so the case of $F$ constant is in \cite{cdg:11}. If $F$ is a non-constant function, one has, from the equation above Proposition 3.4 (page 638) of \cite{cdg:11}:
$$
\Big|\mathbb{E}\Big[\frac{(\gamma_n^{N_{0:n}})^{\otimes 2}(F)}{\gamma_n(1)^2}\Big] - 
\eta_n^{\otimes 2}(F)\Big| = 
$$
$$
\Big|
\sum_{s=1}^{n+1} \sum_{0\leq i_1<\cdots<i_s\leq n}
\Big(\prod_{k=1}^{s}\frac{1}{N_{i_k}}\Big)
\Big(\prod_{k\notin\{i_1,\dots,i_s\}}\big(1-\frac{1}{N_k}\big)
\Big)\overline{\Gamma}_n^{i_1:i_s}(F-\eta_n^{\otimes 2}(F))
\Big|.
$$
Following the proof of Theorem 5.1 of \cite{cdg:11} and noting that one can allow
the function in that paper to be negative, it follows that
$$
|\overline{\Gamma}_n^{i_1:i_s}(F-\eta_n^{\otimes 2}(F))| \leq 
\|F-\eta_n^{\otimes 2}(F)\|_{\infty} \Big(\rho \frac{\overline{C}}{\underline{C}}\Big)^s
\leq 2\|F\|_{\infty}\Big(\rho \frac{\overline{C}}{\underline{C}}\Big)^s.
$$
Thus one has
$$
\Big|\mathbb{E}\Big[\frac{(\gamma_n^{N_{0:n}})^{\otimes 2}(F)}{\gamma_n(1)^2}\Big] - 
\eta_n^{\otimes 2}(F)\Big| \leq
2\|F\|_{\infty}\sum_{s=1}^{n+1} \sum_{0\leq i_1<\cdots<i_s\leq n}
\Big(\prod_{k=1}^{s}\frac{1}{N_{i_k}}\Big)
\Big(\rho \frac{\overline{C}}{\underline{C}}\Big)^s.
$$
Note that 
$$
\sum_{s=1}^{n+1} \sum_{0\leq i_1<\cdots<i_s\leq n}
\Big(\prod_{k=1}^{s}\frac{1}{N_{i_k}}\Big)
\Big(\rho \frac{\overline{C}}{\underline{C}}\Big)^s = 
\prod_{s=0}^n \Big(1+\rho \frac{\overline{C}}{\underline{C}}\frac{1}{N_s}\Big) - 1 \, ,
$$
and for $N_0 > C(n+1),\dots,N_n > C(n+1)$
$$
\prod_{s=0}^n \Big(1+\rho \frac{\overline{C}}{\underline{C}}\frac{1}{N_s}\Big) - 1 \leq
2\rho \frac{\overline{C}}{\underline{C}}\sum_{p=0}^n \frac{1}{N_p} \, ,
$$
(see for instance the proofs of Theorem 5.1 and Corollary 5.2 of \cite{cdg:11}).
It follows that
$$
\Big|
\mathbb{E}\Big[\frac{(\gamma_n^{N_{0:n}})^{\otimes 2}(F)}{\gamma_n(1)^2}\Big] - 
\eta_n^{\otimes 2}(F)
\Big| \leq 8C\|F\|_{\infty}\sum_{p=0}^n \frac{1}{N_p} \, ,
$$
with $C=\rho \frac{\overline{C}}{\underline{C}}$; the proof is concluded.
\end{proof}

\begin{proof}[Proof of Theorem \ref{theo:2}]
Throughout the proof $C<+\infty$ is a constant whose value may change from line-to-line.
It will not depend on the level index.
By \cite[Proposition 2.3]{schweizer:12}
\begin{equation}
\mathbb{E}\Big[\Big(\frac{\gamma_L^{N_{0:L-1}}(1)}{\gamma_L(1)} - 1 \Big)^2\Big] = 
\sum_{p=0}^{L-1}\frac{1}{N_p}\mathbb{E}[T_{p,L}^{N_{0:p}}]\label{eq:schweizer}
\end{equation}
where
$$
T_{p,L}^{N_{0:p}} = \Big(\frac{\gamma_{p}^{N_{0:p-1}}(1)}{\gamma_{p}(1)}\Big)^2
\Big(\eta_p^{N_p}(h_{p,L}^2) - \eta_p^{N_p}(h_{p,L})^2 + 
\eta_p^{N_p}(h_{p,L}) \eta_p^{N_p}\big(\frac{G_p}{\eta_p(G_p)}-1\big)
\Big)
$$
and we use the short-hand for $0\leq p \leq n$, $x\in E$:
$$
h_{p,n}(x) = \frac{Q_{p,n}(1)(x)}{\prod_{q=p}^{n-1}\eta_q(G_q)}.
$$
Now, one has almost surely that
$$
T_{p,L}^{N_{0:p}} = \Big(\frac{\gamma_{p}^{N_{0:p-1}}(1)}{\gamma_{p}(1)}\Big)^2 \times
$$
$$
\Big(\eta_p^{N_p}([h_{p,L}-1]^2) - \eta_p^{N_p}(h_{p,L}-1)^2 + 
\eta_p^{N_p}(h_{p,L}) \eta_p^{N_p}\big(\frac{G_p}{\eta_p(G_p)}-1\big)
\Big).
$$
As 
\begin{eqnarray}
\mathbb{E}[T_{p,L}^{N_{0:p}}] & = & \mathbb{E}\Bigg[
\Big(\frac{\gamma_{p}^{N_{0:p-1}}(1)}{\gamma_{p}(1)}\Big)^2 
\Big(\eta_p^{N_p}([h_{p,L}-1]^2) - \eta_p^{N_p}(h_{p,L}-1)^2\Big)
\Bigg] + \nonumber
\\
& &
\mathbb{E}\Bigg[
\Big(\frac{\gamma_{p}^{N_{0:p-1}}(1)}{\gamma_{p}(1)}\Big)^2 
\eta_p^{N_p}(h_{p,L}) \eta_p^{N_p}\big(\frac{G_p}{\eta_p(G_p)}-1\big)
\Bigg]\label{eq:theo1_main}
\end{eqnarray}
we will consider controlling the two terms on the R.H.S.~of \eqref{eq:theo1_main}
separately.

\noindent\textbf{First term on the R.H.S.~of} \eqref{eq:theo1_main}.\\
We have, almost surely that
\begin{eqnarray*}
\eta_p^{N_p}([h_{p,L}-1]^2) - \eta_p^{N_p}(h_{p,L}-1)^2 & \leq &
C\|h_{p,L}-1\|_{\infty}^2 \\
& \leq & C\Big(\sum_{q=p}^{L-1}\Big\|\frac{G_q}{\eta_p(G_q)}-1\Big\|_{\infty}\Big)^2
\end{eqnarray*}
where we have applied Lemma \ref{lem:new} 
to go to the second line.
Then by Proposition \ref{prop:1} as $N_0 > cL,\dots,N_{L-1} > cL$
$$
\mathbb{E}\Big[\Big(\frac{\gamma_{p}^{N_{0:p-1}}(1)}{\gamma_{p}(1)}\Big)^2\Big] \leq C.
$$
So we have shown that
$$
\mathbb{E}\Bigg[
\Big(\frac{\gamma_{p}^{N_{0:p-1}}(1)}{\gamma_{p}(1)}\Big)^2 
\Big(\eta_p^{N_p}([h_{p,L}-1]^2) - \eta_p^{N_p}(h_{p,L}-1)^2\Big)
\Bigg] \leq
$$
\begin{equation}
C\Big(\sum_{q=p}^{L-1}\Big\|\frac{G_q}{\eta_p(G_q)}-1\Big\|_{\infty}\Big)^2.
\label{eq:theo1_supp1}
\end{equation}

\noindent\textbf{Second term on the R.H.S.~of} \eqref{eq:theo1_main}.\\
We have  almost surely that
$$
\Big(\frac{\gamma_{p}^{N_{0:p-1}}(1)}{\gamma_{p}(1)}\Big)^2 
\eta_p^{N_p}(h_{p,L}) \eta_p^{N_p}\big(\frac{G_p}{\eta_p(G_p)}-1\big)
 = 
$$
$$
\frac{1}{\gamma_p(1)^2}
(\gamma_p^{N_{0:p}})^{\otimes 2}\Big(h_{p,L}\otimes \Big(\frac{G_p}{\eta_p(G_p)}-1\Big)\Big)
$$
and note that, 
$$
\eta_p^{\otimes 2}\Big(h_{p,L}\otimes \Big(\frac{G_p}{\eta_p(G_p)}-1\Big)\Big) = 0.
$$
So by Proposition \ref{prop:1} as $N_0 > cL,\dots,N_{L-1} > cL$ and \eqref{eq:h_cont}
\begin{equation}
\Bigg|\mathbb{E}\Bigg[
\Big(\frac{\gamma_{p}^{N_{0:p-1}}(1)}{\gamma_{p}(1)}\Big)^2 
\eta_p^{N_p}(h_{p,L}) \eta_p^{N_p}\big(\frac{G_p}{\eta_p(G_p)}-1\big)\Bigg]\Bigg|
\leq C\Big\|\frac{G_p}{\eta_p(G_p)}-1\Big\|_{\infty}\frac{(p+1)}{N_p}.
\label{eq:theo1_supp2}
\end{equation}
Combining \eqref{eq:schweizer} with \eqref{eq:theo1_main} and after applying the triangular inequality, the bounds \eqref{eq:theo1_supp1} and 
\eqref{eq:theo1_supp2} complete the proof.

\end{proof}

\section{Proofs for Theorem \ref{theo:1}}
\label{app:proof}

Some of the proofs in this Section will use Proposition \ref{prop:1} in Appendix \ref{app:proof_new}.

\begin{lemma}\label{lem:1}
Let $n\geq 1$ and $f_1,f_2\in\mathcal{B}_b(E)$ then
$$
\mathbb{E}\Big[[\gamma_{n}^{N_{0:n}}-\gamma_{n}](f_1)[\gamma_{n}^{N_{0:n}}-\gamma_{n}](f_2)\Big] = \mathbb{E}[(\gamma_n^{N_{0:n}})^{\otimes 2}(f_1\otimes f_2)] - \gamma_n(1)^2\eta_n^{\otimes 2}(f_1\otimes f_2).
$$
\end{lemma}

\begin{proof}
We have
$$
\mathbb{E}\Big[[\gamma_{n}^{N_{0:n}}-\gamma_{n}](f_1)[\gamma_{n}^{N_{0:n}}-\gamma_{n}](f_2)\Big]
=
$$
$$
\mathbb{E}[\gamma_{n}^{N_{0:n}}(f_1)\gamma_{n}^{N_{0:n}}(f_2)] - \gamma_{n}(f_2)\mathbb{E}[\gamma_{n}^{N_{0:n}}(f_1)] - 
\gamma_{n}(f_1)\mathbb{E}[\gamma_{n}^{N_{0:n}}(f_2)] + 
\gamma_{n}(f_1)\gamma_{n}(f_2) = 
$$
$$
\mathbb{E}[\gamma_{n}^{N_{0:n}}(f_1)\gamma_{n}^{N_{0:n}}(f_2)] - 
\gamma_{n}(f_2)\gamma_{n}(f_1) - \gamma_{n}(f_1)\gamma_{n}(f_2)
+ \gamma_{n}(f_1)\gamma_{n}(f_2)
$$
where the unbiased property of the normalizing constant has been used to go to the last
line. Then it follows that 
$$
\mathbb{E}[\gamma_{n}^{N_{0:n}}(f_1)\gamma_{n}^{N_{0:n}}(f_2)] - 
\gamma_{n}(f_2)\gamma_{n}(f_1) - \gamma_{n}(f_1)\gamma_{n}(f_2)
+ \gamma_{n}(f_1)\gamma_{n}(f_2) = 
$$
$$
\mathbb{E}[(\gamma_n^{N_{0:n}})^{\otimes 2}(f_1\otimes f_2)] - \gamma_n(1)^2\eta_n^{\otimes 2}(f_1\otimes f_2)
$$
which concludes the proof.
\end{proof}

\begin{lemma}\label{lem:2}
Assume (A\ref{hyp:A}-\ref{hyp:B}). Then there exists a $C<+\infty$ such that for any
$2\leq q <p$, $N_0 \geq N_1 \geq \dots \geq N_{q-2} > C(q-1)$:
$$
|\mathbb{E}[[\gamma_{p-2}^{N_{0:p-2}}-\gamma_{p-2}](G_{p-2}(G_{p-1}-1))
[\gamma_{q-2}^{N_{0:q-2}}-\gamma_{q-2}](G_{q-2}(G_{q-1}-1))]| \leq 
$$
$$
\frac{c(q-1)\gamma_{q-2}(1)^2}{N_{q-2}}
\Big\|G_{q-2}(G_{q-1}-1)
Q_{q-2,p-2}(G_{p-2}(G_{p-1}-1))
\Big\|_{\infty}.
$$
\end{lemma}

\begin{proof}
From \cite[Proposition 7.4.1]{delm:04} we have
$$
\mathbb{E}[[\gamma_{p-2}^{N_{0:p-2}}-\gamma_{p-2}](G_{p-2}(G_{p-1}-1))
[\gamma_{q-2}^{N_{0:q-2}}-\gamma_{q-2}](G_{q-2}(G_{q-1}-1))] = 
$$
$$
\sum_{s_1=0}^{p-2}\sum_{s_2=0}^{q-2}
\mathbb{E}\Big[
\gamma_{s_1}^{N_{0:s_1-1}}(1)[\eta_{s_1}^{N_{s_1}}-\Phi_{s_1}(\eta_{s_1-1}^{N_{s_1-1}})]
(Q_{s_1,p-2}(\overline{G}_p))\times
$$
$$
\gamma_{s_2}^{N_{0:s_2-1}}(1)[\eta_{s_2}^{N_{s_2}}-\Phi_{s_2}(\eta_{s_2-1}^{N_{s_2-1}})]
(Q_{s_2,q-2}(\overline{G}_q))
\Big]
$$
where we have used the shorthand $\overline{G}_{s} = G_{s-2}(G_{s-1}-1)$ for any $s\geq 2$.
For any $s\geq 0$, $f\in\mathcal{B}_b(E)$
$$
\mathbb{E}[\gamma_{s}^{N_{0:s-1}}(1)[\eta_{s}^{N_{s}}-\Phi_{s}(\eta_{s-1}^{N_{s-1}})]
(f)|\mathscr{F}_{s-1}^{N_{0:s-1}}] = 0
$$
thus, it follows that
$$
\mathbb{E}[[\gamma_{p-2}^{N_{0:p-2}}-\gamma_{p-2}](G_{p-2}(G_{p-1}-1))
[\gamma_{q-2}^{N_{0:q-2}}-\gamma_{q-2}](G_{q-2}(G_{q-1}-1))] = 
$$
$$
\sum_{s=0}^{q-2}
\mathbb{E}[\gamma_{s}^{N_{0:s-1}}(1)^2[\eta_{s}^{N_{s}}-\Phi_{s}(\eta_{s-1}^{N_{s-1}})]
^{\otimes 2}
(Q_{s,p-2}(\overline{G}_p)\otimes Q_{s,q-2}(\overline{G}_q))].
$$
Now
for any $n\geq 1$, $f_1,f_2\in\mathcal{B}_b(E)$, one can show,
using almost the same calculations as above, that the following holds
$$
\sum_{s=0}^{n}
\mathbb{E}[\gamma_{s}^{N_{0:s-1}}(1)^2[\eta_{s}^{N_{s}}-\Phi_{s}(\eta_{s-1}^{N_{s-1}})]
^{\otimes 2}
(Q_{s,n}(f_1)\otimes Q_{s,n}(f_2))] =
$$
$$
\mathbb{E}\Big[[\gamma_{n}^{N_{0:n}}-\gamma_{n}](f_1)[\gamma_{n}^{N_{0:n}}-\gamma_{n}](f_2)\Big].
$$
Using this equality with $n=q-2$, and the fact that $Q_{s,p-2} = Q_{s,q-2} Q_{q-2,p-2}$, finally
$$
\mathbb{E}[[\gamma_{p-2}^{N_{0:p-2}}-\gamma_{p-2}](G_{p-2}(G_{p-1}-1))
[\gamma_{q-2}^{N_{0:q-2}}-\gamma_{q-2}](G_{q-2}(G_{q-1}-1))] = 
$$
$$
\mathbb{E}\Big[[\gamma_{q-2}^{N_{0:n}}-\gamma_{q-2}](Q_{q-2,p-2}(\overline{G}_p))[\gamma_{q-2}^{N_{0:n}}-\gamma_{q-2}](\overline{G}_q)\Big].
$$
Then, by Lemma \ref{lem:1}:
$$
\mathbb{E}[[\gamma_{p-2}^{N_{0:p-2}}-\gamma_{p-2}](G_{p-2}(G_{p-1}-1))
[\gamma_{q-2}^{N_{0:q-2}}-\gamma_{q-2}](G_{q-2}(G_{q-1}-1))] = 
$$
$$
\mathbb{E}[(\gamma_{q-2}^{N_{0:q-2}})^{\otimes 2}(Q_{q-2,p-2}(\overline{G}_p)\otimes \overline{G}_q)] - \gamma_{q-2}(1)^2\eta_{q-2}^{\otimes 2}(Q_{q-2,p-2}(\overline{G}_p)\otimes \overline{G}_q).
$$
Then, one can apply Proposition \ref{prop:1} to obtain that
$$
|\mathbb{E}[[\gamma_{p-2}^{N_{0:p-2}}-\gamma_{p-2}](G_{p-2}(G_{p-1}-1))
[\gamma_{q-2}^{N_{0:q-2}}-\gamma_{q-2}](G_{q-2}(G_{q-1}-1))]| \leq 
$$
$$
\frac{C(q-1)\gamma_{q-2}(1)^2}{N_{q-2}}
\Big\|G_{q-2}(G_{q-1}-1)
Q_{q-2,p-2}(G_{p-2}(G_{p-1}-1))
\Big\|_{\infty}.
$$
\end{proof}

\begin{proof}[Proof of Theorem \ref{theo:1}]
Throughout the proof $C<+\infty$ is a constant whose value may change from line-to-line.
It will not depend on the level index.
We have
$$
\mathbb{E}\Big[\Big(\frac{\tilde{\gamma}_L^{N_{0:L-2}}(1)}{\gamma_L(1)}-1\Big)^2\Big]
\leq
$$
$$
\frac{1}{\gamma_L(1)^2}\mathbb{E}[[\eta_0^{N_0}-\eta_0](G_0)^2]
+ \frac{1}{\gamma_L(1)^2}
\mathbb{E}[(\sum_{p=2}^L[\gamma_{p-2}^{N_{0:p-2}}-\gamma_{p-2}](\overline{G}_{p}) )^2].
$$
As $\gamma_L(1) = Z_L/Z_0\geq \underline{C}/\overline{C}$ it follows by standard results for i.i.d.~random variables that one has
$$
\frac{1}{\gamma_L(1)^2}\mathbb{E}[[\eta_0^{N_0}-\eta_0](G_0)^2] \leq 
\frac{C}{N_0}.
$$

Now
$$
\mathbb{E}[(\sum_{p=2}^L[\gamma_{p-2}^{N_{0:p-2}}-\gamma_{p-2}](\overline{G}_{p}) )^2] =
\sum_{p=2}^L\gamma_{p-2}(1)^2\mathbb{E}\Big[\frac{\gamma_{p-2}^{N_{0:p-2}}(\overline{G}_p)^2}{\gamma_{p-2}(1)^2}-\eta_{p-2}(\overline{G}_{p})^2\Big]
$$
$$
+ 2\sum_{p=2}^L\sum_{q=2}^{p-1}
\mathbb{E}[[\gamma_{p-2}^{N_{0:p-2}}-\gamma_{p-2}](\overline{G}_p)
[\gamma_{q-2}^{N_{0:q-2}}-\gamma_{q-2}](\overline{G}_q)].
$$
Applying Propositon \ref{prop:1} to the terms in the single sum
and Lemma \ref{lem:2} to the terms in the double sum, we have that
$$
\mathbb{E}[(\sum_{p=2}^L[\gamma_{p-2}^{N_{0:p-2}}-\gamma_{p-2}](\overline{G}_{p}) )^2] \leq
C\Big(\sum_{p=2}^L \gamma_{p-2}(1)^2\frac{(p-1)}{N_p}\|\overline{G}_p\|_{\infty}^2
$$
$$
+ \sum_{p=2}^L\sum_{q=2}^{p-1}\frac{(q-1)\gamma_{q-2}(1)^2}{N_{q-2}}
\Big\|G_{q-2}(G_{q-1}-1)
Q_{q-2,p-2}(G_{p-2}(G_{p-1}-1))
\Big\|_{\infty}
\Big).
$$
As $\gamma_{p-2}(1)\leq \overline{C}/\underline{C}$, $\gamma_L(1)\geq \underline{C}/\overline{C}$ one has
$$
\frac{1}{\gamma_L(1)^2}\sum_{p=2}^L \gamma_{p-2}(1)^2\frac{(p-1)}{N_p}\|\overline{G}_p\|^2_{\infty}
\leq C \sum_{p=2}^L \frac{(p-1)}{N_{p-2}}\|G_{p-1}-1\|_{\infty}^2.
$$
We have 
$$
\frac{1}{\gamma_L(1)^2}\sum_{p=2}^L\sum_{q=2}^{p-1}\frac{(q-1)\gamma_{q-2}(1)^2}{N_{q-2}}
\Big\|G_{q-2}(G_{q-1}-1)
Q_{q-2,p-2}(G_{p-2}(G_{p-1}-1))
\Big\|_{\infty} =
$$
$$
\sum_{p=2}^L\sum_{q=2}^{p-1}\frac{c(q-1)\gamma_{q-2}(1)}{\gamma_L(1)N_{q-2}}
\frac{1}{\eta_{q-2}(Q_{q-2,p-2}(1))}\frac{Z_{p-1}}{Z_L}\times
$$
$$
\Big\|G_{q-2}(G_{q-1}-1)
Q_{q-2,p-2}(G_{p-2}(G_{p-1}-1))
\Big\|_{\infty}.
$$ 
Then as $\gamma_{q-2}(1)\leq \overline{C}/\underline{C}$, $\gamma_L(1)\geq \underline{C}/\overline{C}$, $Z_{p-1}\leq C$, $Z_L\geq C$ and by \cite[Lemma 4.1]{cdg:11}
$$
\frac{Q_{q-2,p-2}(G_{p-2}(G_{p-1}-1))}{\eta_{q-2}(Q_{q-2,p-2}(1))} \leq C
\|G_{p-1}-1\|_{\infty}
$$
we have 
$$
\frac{1}{\gamma_L(1)^2}\sum_{p=2}^L\sum_{q=2}^{p-1}\frac{(q-1)\gamma_{q-2}(1)^2}{N_{q-2}}
\Big\|G_{q-2}(G_{q-1}-1)
Q_{q-2,p-2}(G_{p-2}(G_{p-1}-1))
\Big\|_{\infty}\leq
$$
$$
C\sum_{p=2}^L \sum_{q=2}^{p-1} \frac{(q-1)}{N_{q-2}}
\|G_{p-1}-1\|_{\infty}\|G_{q-1}-1\|_{\infty}.
$$
From here one can easily conclude.
\end{proof}

\end{document}